\documentclass[12pt]{article}
\newenvironment{proof}
{\pagebreak[1]{\narrower\noindent {\bf Proof:\quad\nopagebreak}}}{\QED}

\usepackage{amsmath}
\usepackage{amssymb}
\usepackage{url}
\usepackage{graphics}
\usepackage{enumerate}

\newcommand{\ang}[1]{\langle#1\rangle}

\newcommand{\xvec}[1]{\ifcase 3{#1} {\ang {x_1,x_2,x_3} } \else
\ifcase 4{#1} {\ang{x_1,x_2,x_3,x_4}} \else {\ang {x_1,\ldots,x_{#1}}}\fi\fi}
\newcommand{\yvec}[1]{\ifcase 3{#1} {\ang {y_1,y_2,y_3} } \else
\ifcase 4{#1} {\ang{y_1,y_2,y_3,y_4}} \else {\ang {y_1,\ldots,y_{#1}}}\fi\fi}
\newcommand{\zvec}[1]{\ifcase 3{#1} {\ang {z_1,z_2,z_3} } \else
\ifcase 4{#1} {\ang{z_1,z_2,z_3,z_4}} \else {\ang {z_1,\ldots,z_{#1}}}\fi\fi}
\newcommand{\vecc}[2]{\ifcase 3{#2} {\ang { {#1}_1,{#1}_2,{#1}_3 } } \else
\ifcase 4{#1} {\ang { {#1}_1,{#1}_2,{#1}_3,{#1}_{4} } }
\else {\ang { {#1}_1,\ldots,{#1}_{#2}}}\fi\fi}
\newcommand{\veccd}[3]{\ifcase 3{#2} {\ang { {#1}_{{#3}1},{#1}_{{#3}2},{#1}_{{#3}3} } } \else
\ifcase 4{#1} {\ang { {#1}_{{#3}1},{#1}_{{#3}2},{#1}_{#3}3},{#1}_{{#3}4} }
\else {\ang { {#1}_{{#3}1},\ldots,{#1}_{{#3}{#2}}}}\fi\fi}
\newcommand{\veccz}[2]{\ifcase 3{#2} {\ang { {#1}_0,{#1}_2,{#1}_3 } } \else
\ifcase 4{#1} {\ang { {#1}_0,{#1}_2,{#1}_3,{#1}_{4} } }
\else {\ang { {#1}_0,\ldots,{#1}_{#2}}}\fi\fi}
\newcommand{\xve}[1]{\ifcase 3{#1} {x_1,x_2,x_3} \else
\ifcase 4{#1} {x_1,x_2,x_3,x_4} \else {x_1,\ldots,x_{#1}}\fi\fi}
\newcommand{\yve}[1]{\ifcase 3{#1} {y_1,y_2,y_3} \else
\ifcase 4{#1} {y_1,y_2,y_3,y_4} \else {y_1,\ldots,y_{#1}}\fi\fi}
\newcommand{\zve}[1]{\ifcase 3{#1} {z_1,z_2,z_3} \else
\ifcase 4{#1} {z_1,z_2,z_3,z_4} \else {z_1,\ldots,z_{#1}}\fi\fi}
\newcommand{\ve}[2]{\ifcase 3#2 {{#1}_1,{#1}_2,{#1}_3} \else
\ifcase 4#2 {{#1}_1,{#1}_2,{#1}_3,{#1}_{4}}
\else {{#1}_1,\ldots,{#1}_{#2}}\fi\fi}
\newcommand{\ved}[3]{\ifcase 3#2 {{#1}_{{#3}1},{#1}_{{#3}2},{#1}_{{#3}3}} \else
\ifcase 4#2 {{#1}_{{#3}1},{#1}_{{#3}2},{#1}_{{#3}3},{#1}_{{#3}4}}
\else {{#1}_{{#3}1},\ldots,{#1}_{{#3}{#2}}}\fi\fi}
\newcommand{\fuve}[3]{
\ifcase 3#2
{{#3}({#1}_1),{#3}({#1}_2,{#3}({#1}_3)} \else
\ifcase 4#2
{{#3}({#1}_1),{#3}({#1}_2),{#3}({#1}_3),{#3}({#1}_4)}
\else
{{#3}({#1}_1),\ldots,{#3}({#1}_{#2})}\fi\fi}
\newcommand{\setmathchar}[1]{\ifmmode#1\else$#1$\fi}
\newcommand{\vlist}[2]{%
	\setmathchar{%
		\compound#2\one{#2}\two
		\ifcompound
			({#1}_1,\ldots,{#1}_{#2})
		\else
			\ifcat N#2
				({#1}_1,\ldots,{#1}_{#2})
			\else
				\ifcase#2
					({#1}_0)\or
					({#1}_1)\or
					({#1}_1,{#1}_2)\or
					({#1}_1,{#1}_2,{#1}_3)\or
					({#1}_1,{#1}_2,{#1}_3,{#1}_4)\else
					({#1}_1,\ldots,{#1}_{#2})
				\fi
			\fi
		\fi}}

\newif\ifcompound
\def\compound#1\one#2\two{%
	\def\one{#1}
	\def\two{#2}
	\if\one\two
		\compoundfalse
	\else
		\compoundtrue
	\fi}

\newcommand{\xwe}[1]{\ifcase 3{#1} {x_1\wedge x_2\wedge x_3} \else
\ifcase 4{#1} {x_1\wedge x_2\wedge x_3\wedge x_4} \else {x_1\wedge\citedots \wedge
x_{#1}}\fi\fi}
\newcommand{\we}[2]{\ifcase 3#2 {\ang { {#1}_1\wedge {#1}_2\wedge {#1}_3 } } \else
\ifcase 4{#1} {\ang { {#1}_1\wedge {#1}_2\wedge {#1}_3\wedge {#1}_{4} } }
\else {\ang { {#1}_1\wedge\citedots\wedge {#1}_{#2}}}\fi\fi}
\newcommand{\s}[1]{\s_{#1}}

\newcommand{\monus}{\;\raise.5ex\hbox{{${\buildrel
    \ldotp\over{\hbox to 6pt{\hrulefill}}}$}}\;}

\newcounter{savenumi}

\newtheorem{theoremfoo}{Theorem}[section] 
\newenvironment{theorem}{\pagebreak[1]\begin{theoremfoo}}{\end{theoremfoo}}

\newtheorem{lemmafoo}[theoremfoo]{Lemma}

\newtheorem{conjecturefoo}[theoremfoo]{Conjecture}

\newenvironment{conjecture}{\pagebreak[1]\begin{conjecturefoo}}{\end{conjecturefoo}}

\newtheorem{conventionfoo}[theoremfoo]{Convention}

\newtheorem{porismfoo}[theoremfoo]{Porism}

\newtheorem{gamefoo}[theoremfoo]{Game}

\newtheorem{corollaryfoo}[theoremfoo]{Corollary}
\newenvironment{corollary}{\pagebreak[1]\begin{corollaryfoo}}{\end{corollaryfoo}}

\newtheorem{claimfoo}[theoremfoo]{Claim}

\newtheorem{openfoo}[theoremfoo]{Open Problem}

\newtheorem{exercisefoo}{Exercise}

\newcommand{\fig}[1] 
{
 \begin{figure}
 \begin{center}
 \input{#1}
 \end{center}
 \end{figure}
}

\newtheorem{potanafoo}[theoremfoo]{Potential Analogue}

\newtheorem{notefoo}[theoremfoo]{Note}

\newtheorem{notabenefoo}[theoremfoo]{Nota Bene}

\newtheorem{nttn}[theoremfoo]{Notation}

\newtheorem{empttn}[theoremfoo]{Empirical Note}

\newtheorem{examfoo}[theoremfoo]{Example}

\newtheorem{dfntn}[theoremfoo]{Def}

\newtheorem{propositionfoo}[theoremfoo]{Proposition}

\newcommand{\yyskip}{\penalty-50\vskip 5pt plus 3pt minus 2pt}
\newcommand{\blackslug}{\hbox{\hskip 1pt
        \vrule width 4pt height 8pt depth 1.5pt\hskip 1pt}}
\newcommand{\QED}{{\penalty10000\parindent 0pt\penalty10000
        \hskip 8 pt\nolinebreak\blackslug\hfill\lower 8.5pt\null}
        \par\yyskip\pagebreak[1]}

\newcommand{\BBB}{{\penalty10000\parindent 0pt\penalty10000
        \hskip 8 pt\nolinebreak\hbox{\ }\hfill\lower 8.5pt\null}
        \par\yyskip\pagebreak[1]}

\newtheorem{factfoo}[theoremfoo]{Fact}

\newenvironment{block}{\begin{list}{\hbox{}}{\leftmargin 1em
    \itemindent -1em \topsep 0pt \itemsep 0pt \partopsep 0pt}}{\end{list}}

\usepackage{enumitem}
\usepackage{xcolor}
\usepackage[all]{xy}
\usepackage{wrapfig}

\title{Average-Case Hardness of Proving Tautologies and Theorems}
\author{Hunter Monroe}
\date{\today}
\begin{document}
\maketitle

\begin{abstract}
We consolidate two widely believed conjectures about tautologies---no optimal proof system exists, and most require superpolynomial size proofs in any system---into a $p$-isomorphism-invariant condition satisfied by all paddable $\textbf{coNP}$-complete languages or none. The condition is: for any Turing machine (TM) $M$ accepting the language, $\textbf{P}$-uniform input families requiring superpolynomial time by $M$ exist (equivalent to the first conjecture) and appear with positive upper density in an enumeration of input families (implies the second). In that case, no such language is easy on average (in $\textbf{AvgP}$) for a distribution applying non-negligible weight to the hard families.

The hardness of proving tautologies and theorems is likely related. Motivated by the fact that arithmetic sentences encoding ``string $x$ is Kolmogorov random'' are true but unprovable with positive density in a finitely axiomatized theory $\mathcal{T}$ (Calude and J{\"u}rgensen), we conjecture that any propositional proof system requires superpolynomial size proofs for a dense set of $\textbf{P}$-uniform families of tautologies encoding ``there is no $\mathcal{T}$ proof of size $\leq t$ showing that string $x$ is Kolmogorov random''. This implies the above condition.

The conjecture suggests that there is no optimal proof system because undecidable theories help prove tautologies and do so more efficiently as axioms are added, and that constructing hard tautologies seems difficult because it is impossible to construct Kolmogorov random strings. Similar conjectures that computational blind spots are manifestations of noncomputability would resolve other open problems.
\end{abstract}
\section{Introduction}
The literature has explored the possibility that any TM accepting a given $\textbf{coNP}$-complete language requires superpolynomial time on some $\textbf{P}$-uniform input family (a \emph{hard sequence}).\footnote{Comments are appreciated from Scott Aaronson, Valentina Harizonov, Mehmet Kayaalp, Daniel Monroe, Pavel Pudl{\'a}k, Luca Trevisan, and participants in seminars at George Washington University and Davidson College. This paper originated from a discussion with Russell Impagliazzo on formalizing the idea that proof speedup for tautologies and arithmetic were the same phenomenon. Remaining errors are my own.}  This paper considers a stronger conjecture that such hard sequences appear with positive upper density in an enumeration of input families (\emph{dense hard sequences}, precise definitions will follow). It finds that if any paddable $\textbf{coNP}$-complete language has dense hard sequences, they all do, and none are in $\textbf{AvgP}$. It also explores whether proving tautologies and theorems is hard on average, and whether these two are related.\footnote{The phrase ``hardness of proving tautologies'' below refers only to proof size and not the hardness of finding proofs.} It also considers the hypothesis that computational hardness is a manifestation of noncomputability, by formalizing the intuition that resource-bounded Turing machines (TMs) have blind spots regarding noncomputable languages.
	
The existence of hard sequences for a $\textbf{coNP}$-complete language $L$ has broad implications, as shown by a rich literature.\footnote{This literature was launched by Kraj\'{\i}\v{c}ek and Pudl{\'a}k\cite{Krajicek}. Hard sequences for TMs and for proof systems have been examined by Chen et al\cite{ChenFlumMuller}, Kraj\'{\i}\v{c}ek\cite{Krajicekbounded}\cite{Krajicekproof}, and Monroe\cite{MonroeTCS}. The definition of almost optimal comes from \cite{Krajicek}, but that terminology is from Chen and Flum\cite{ChenFlumLogicsPTIME}.} To define terms, let  $(x_s)_{s\in \mathbb{N}}$ be a sequence of inputs $x_s\in L$ such that some TM $A$ on input $1^s$ produces $x_s$ in at most $s^c$ steps, for some constant $c$. Say that $(x_s)_{s\in \mathbb{N}}$ is a \emph{hard sequence} for $M$ if $(T_M(x_s))_{s\in \mathbb{N}}$ is not polynomially bounded, where $T_M$ is the function that maps a string $x$ to how many steps $M(x)$ takes. If every $M$ accepting $L$ has hard sequences, say that $L$ has hard sequences. In that case, $L$ cannot have a best algorithm $M$, because we can create $M'$ that checks whether an element of the hard sequence appears as the input and runs $M$ otherwise. In that case, say that $L$ has no \emph{almost optimal} $M$, that is, it is not the case that for any other $M'$, there exists $c$, such for all $x$, $T_M(x)\leq (T_{M'}(x)+|x|)^c$.

Hard sequences can also be defined for propositional proof systems for tautologies ($\texttt{TAUT}$). A propositional proof system is a function $h\in \textbf{FP}$ with range $\texttt{TAUT}$.\cite{CookReckhow} For tautology $\tau$, any string $w$ such that $h(w)=\tau$ is a proof of $\tau$. The proof system $h$ is $p$-optimal if for any other proof system $f$, there exists $g\in \textbf{FP}$ such that $h(g(x))=f(x)$.\cite{Krajicek} A hard sequence of tautologies for a proof system is defined the same as for TMs but with the requirement that tautologies have proofs of superpolynomial size. The concepts coincide for a nondeterministic TM that accepts $\texttt{TAUT}$ with the accepting path as the proof.

The existence of hard sequences for a paddable $\textbf{coNP}$-complete language is one of a large group of equivalent statements listed in the following theorem. Statement (i) (and therefore the others) is widely believed to hold, while statement (ix) will play a key role in this paper.

\begin{theorem}\label{equivalentstatements}
The following statements are equivalent:
\begin{enumerate}[label=\roman*.]
	\item There is no $p$-optimal propositional proof system for $\texttt{TAUT}$.\cite{Krajicek}
\item There exists a paddable $\textbf{coNP}$-complete language with hard sequences.\cite{ChenFlumMuller}
\item All paddable $\textbf{coNP}$-complete languages have hard sequences.\cite{ChenFlumMuller} 
\item There exists a paddable $\textbf{coNP}$-complete language with no almost optimal algorithm.\cite{Krajicek} 
\item For all paddable $\textbf{coNP}$-complete languages, there is no almost optimal algorithm.\cite{Krajicek} 
\item For any propositional proof system, some hard sequence of tautologies in $\textbf{P}$ has superpolynomial size proofs.\footnote{See Kraj\'{\i}\v{c}ek\cite{Krajicekbounded} Theorem 14.2.2 and Chen et al\cite{ChenFlumMuller} Theorem 6.7.} 
\item No theory $\mathcal{T}$ proves the soundness of all propositional proof systems.\cite{PudlakSimons}
\item For any $M$ accepting $\texttt{coBHP}=\{\langle N,x,1^t\rangle|$ there is no accepting path of nondeterministic TM $N$ on input $x$ with $t$ or fewer steps$\}$, there exists some non-halting $\langle N',x'\rangle$ such that $f(t)=T_M(\langle N',x',1^t\rangle)$ is superpolynomial.\cite{MonroeTCS}
\item For any $M$ accepting  $\texttt{coTHEOREMS}=\{\langle\phi,1^t\rangle|$ sentence $\phi$ has no formal proof in theory $\mathcal{T}$ with size $\leq t\}$, there exists some sentence $\phi'$ with no proof in $\mathcal{T}$ of any length, such that $f(t)=T_M(\langle\phi',1^t\rangle)$ is superpolynomial.\footnote{Chen and Flum\cite{ChenFlumParameterized} Theorem 31 that \texttt{coBHP} and \texttt{coTHEOREMS} are fixed parameter tractable reducible to each other, so Monroe\cite{MonroeTCS} applies.}
\end{enumerate}
\end{theorem}

If any of the above statements hold, then $\textbf{P}\neq\textbf{NP}$ and $\textbf{EXP}\neq\textbf{NEXP}$.\cite{Krajicek} The converse is not known to hold. There are nondeterministic analogs of each of the above statements which imply $\textbf{NP}\neq\textbf{coNP}$ and $\textbf{NEXP}\neq\textbf{coNEXP}$. 

Statements (viii) and (ix) are more specific than the others and highlight an apparent linkage with noncomputability. These statements suggest that resource-bounded TMs have blind spots regarding the noncomputable languages ``nonhalting TMs'' and ``unprovable arithmetic sentences''. This paper goes further by exploring whether the blind spots are more severe, relating to a dense set of true unprovable arithmetic sentences.

The paper is organized as follows. Section \ref{hardsequencessection} shows that if one paddable $\textbf{coNP}$-complete language has dense hard sequences, they all do, and none are easy on average. Section \ref{densityunprovablesection} notes that true arithmetic sentences with no proof (of any size) have positive density among length $n$ sentences (Calude and J{\"u}rgensen) and suggests this is connected to the existence of dense hard sequences. Section \ref{tautsection} considers whether the set of tautologies has dense hard sequences. Section \ref{conclusion} concludes and identifies questions for further research. 

\section{Density of Hard Sequences}\label{hardsequencessection}
This section shows that a language is not easy on average if it has dense hard sequences, focusing initially on the example of $\texttt{coTHEOREMS}$ rather than $\texttt{TAUT}$ to highlight the role of hard sequences. Define $\texttt{coTHEOREMS}=\{\langle\phi,1^t\rangle|$ sentence $\phi$ has no formal proof in theory $\mathcal{T}$ with size $\leq t\}$. The theory $\mathcal{T}$ can be any finitely axiomatized, sound, consistent theory sufficiently strong to formalize arithmetic, such as Peano arithmetic (PA) or Zermelo Fraenkel set theory with the axiom of choice (ZFC).

If the statements in Theorem \ref{equivalentstatements} hold, then by statement (ix), there exists some hard sequence for $\texttt{coTHEOREMS}$ of the form $(\langle \phi,1^t\rangle)_{t\in \mathbb{N}}$, where $\phi$ is not provable in $\mathcal{T}$ (for any size proof). A statement of density rather than just existence is as follows. Fix an acceptable numbering of sentences $(\phi_i)$ in $\mathcal{T}$ for $i\in \mathbb{N}$. For TM $M$ accepting \texttt{coTHEOREMS}, let $H_M=\{i|(\langle \phi_i,1^t\rangle)_{t\in \mathbb{N}}$ is a hard sequence for $M\}$. Say that $H_M$ has \emph{positive upper density} if $\displaystyle\limsup_{n \rightarrow\infty}\frac{|H_M\cap\{1,2,\ldots,n\}|}{n}>0$. Denote this limit as $p_M$. Let $p_\texttt{coTHEOREMS}$ be the minimum $p_M$ over all $M$. Say that $\texttt{coTHEOREMS}$ has \emph{dense hard sequences} if for every $M$, $H_M$ has positive upper density, and $p_\texttt{coTHEOREMS}>0$. 

If $\texttt{coTHEOREMS}$ (or any other language) has dense hard sequences, then accepting it is not easy on average (in $\textbf{AvgP}$), with one proviso. Recall that a distributional problem $(L,\mathcal{D})$ (where $\mathcal{D}$ is an ensemble of distributions for length $n$ inputs) is in $\textbf{AvgP}$ if there is a TM $M$ accepting $L$ and constants $C$ and $\epsilon>0$ such for that every $n$: $\displaystyle \mathop{\mathbb{E}}_{y\in_R \mathcal{D}_n} \displaystyle\bigg[ \frac{T_M(y)^\epsilon}{n} \bigg]\leq C$.

A distribution that applied zero or negligible weight to the hard sequences would negate their effect on the average-case hardness of $\texttt{coTHEOREMS}$, so a technical assumption is needed to rule this out. Define a \emph{balanced} distribution as one that applies non-negligible weight to the dense hard sequences for $\texttt{coTHEOREMS}$---the probability weight applied to those length $n$ inputs appearing in a dense hard sequence exceeds some bound $c$ infinitely often. The uniform distribution is an obvious example.\footnote{The universal distribution is also balanced, because it applies at least probability weight $2^{-|x|}$ to any input $x$ and therefore applies non-negligible weight to any potential hard sequence. See Li and Vitanyi\cite{LiVitanyiBook}.} 

Then we have:

\begin{theorem}\label{notinAvgP}
If any language $L$ has dense hard sequences, then  $L\notin\textbf{AvgP}$ for a balanced distribution.
\end{theorem}
\begin{proof}
Fix $M$, $C$ and $\epsilon$, where $M$ accepts $L$ and calculate the expectation considering only inputs $y$ appearing in a hard sequence for $M$. For each sequence, $T_M(y)$ grows faster than  $|y|^k$, in particular for $k$ greater than $1/\epsilon$. Therefore, $\displaystyle \mathop{\mathbb{E}}_{y\in_R \mathcal{D}_n} \bigg[ \frac{ T_M(y)^\epsilon}{n} \bigg]$ is unbounded as $n$ tends to infinity. 
\end{proof}
The converse is not known to hold: the property ``has dense hard sequences'' appears to be stronger than ``is not in $\textbf{AvgP}$''. The literature has similarly noted that the property ``has hard sequences'' appears to be stronger than ``is not in $\textbf{P}$''.

If one paddable $\textbf{coNP}$-complete language has (not necessarily dense) hard sequences, then they all do (Chen et al\cite{ChenFlumMuller}), and none are in $\textbf{P}$. This result follows from the $p$-isomorphism of such languages, and the invariance of the property ``has hard sequences'' under $p$-isomorphism. We prove a similar result: if one paddable $\textbf{coNP}$-complete language has dense hard sequences, then they all do, and none are in $\textbf{AvgP}$ for a balanced distribution. We do so by extending the definitions for these terms from $\texttt{coTHEOREMS}$ to other $\textbf{coNP}$-complete languages so as to ensure invariance under $p$-isomorphism.

Let $L$ be any paddable $\textbf{coNP}$-complete language. A $p$-\emph{isomorphism} between languages $\texttt{coTHEOREMS}$ and $L$ is a bijective map $f:\texttt{coTHEOREMS}\rightarrow L$ such that $x\in\texttt{coTHEOREMS}$ if and only if $f(x)\in L$ and $f$ and its inverse can be computed in time polynomial in their arguments. Berman and Hartmanis\cite{BermanHartmanis} show that any two paddable $\textbf{coNP}$-complete languages are $p$-isomorphic. Thus, the set of paddable $\textbf{coNP}$-complete languages can be seen as a single language with different encoding schemes that take polynomial time to encode and decode. We will take this view, with $\texttt{coTHEOREMS}$ as that single language, and with all other paddable $\textbf{coNP}$-complete languages considered as being $\texttt{coTHEOREMS}$ with a different polynomial time encoding scheme $f$. By doing so, we are effectively treating $\texttt{coTHEOREMS}$ as a canonical form for paddable $\textbf{coNP}$-complete languages with the crucial information about the dense hard sequences.

For a paddable $\textbf{coNP}$-complete language $L$ other than $\texttt{coTHEOREMS}$ with $p$-isomorphism $f$, consider the sequences $(f(\langle\phi_i,1^t\rangle))_{t\in\mathbb{N}}$ in $L$, which will play the role of $(\langle\phi_i,1^t\rangle)_{t\in \mathbb{N}}$ in $\texttt{coTHEOREMS}$.\footnote{The same approach can be applied for $p$-isomorphisms of $\texttt{coTHEOREMS}$ with itself.} Say that $L$ has dense hard sequences if for any $M_L$ accepting $L$, a positive upper density of the sequences $(f(\langle\phi_i,1^t\rangle))_{t\in\mathbb{N}}$ over $i$ are hard sequences. Since $M= M_L\circ f$ is a TM accepting $\texttt{coTHEOREMS}$, this is equivalent to the statement ``$\texttt{coTHEOREMS}$ has dense hard sequences''. Say that a distribution over $L$ is balanced if it applies non-negligible weight to the dense hard sequences $(f(\langle\phi_i,1^t\rangle))_{t\in\mathbb{N}}$. This is crucial---the balanced distributions for $L$ are determined by $f$ and by the structure of dense hard sequences for $\texttt{coTHEOREMS}$.\footnote{To illustrate the role of this assumption, design a $\textbf{coNP}$-complete language $L$ to be easy on average for a uniform distribution. Let $\texttt{HC100}$ be those graphs, with a weight for each edge stored using 100 bits including leading zeros, for which there is a Hamilton Circuit in which each edge has weight one with 99 leading zeros. With uniformly random weights, no edge in a graph will have 99 leading zeros with probability one asymptotically, meaning they will be easily identified as being in $\texttt{coHC100}$, so this $\textbf{coNP}$-complete language would be easy on average with a uniform distribution. A balanced distribution would apply negligible weight to the easy cases for $\texttt{coHC100}$. We appreciate an anonymous Stack Exchange user for suggesting this construction.} Keep in mind that the uniform distribution for $\texttt{coTHEOREMS}$ is always balanced; the next section formulates a conjecture implying a clearer specification of what ``balanced'' means.

\begin{theorem}\label{ifonethenall}
Suppose one paddable $\textbf{coNP}$-complete language has dense hard sequences. Then they all do, and none are easy on average (in $\textbf{AvgP}$) for a balanced distribution.\footnote{It is possible to show with a weaker assumption that all paddable $\textbf{coNP}$- (and $\textbf{NP}$-) complete languages are not in $\textbf{AvgP}$, but with highly skewed distributions (per Luca Trevisan, private communication). If the statements in Theorem \ref{equivalentstatements} hold, then $\textbf{EXP}\neq\textbf{NEXP}$, there is a unary language $L$ in $\textbf{NP/P}$, and $L\notin\textbf{AvgP}$ assuming a distribution with probability one for the $1^n$ string. Suppose $L$ with that distribution has a length non-decreasing reduction to a distributional problem $(A,\mathcal{D})$, where $\mathcal{D}$ is uniform over length $n$ strings of the form $f(1^t)$ for $t\leq n$, if any exist, and applies probability one to $0^n$ otherwise. If $(A,\mathcal{D})\in\textbf{AvgP}$, then there is a worst-case polynomial time algorithm for $L$, which is a contradiction.}
\end{theorem}
\begin{proof}
We have chosen definitions to make this true.
\end{proof}

Theorem \ref{ifonethenall} proves a strong conclusion with little work, by comparison with the literature on average-case completeness for $\textbf{NP}$.\cite{BogdanovTrevisan} Without defining reductions and completeness, the theorem already identifies what are essentially a set of complete languages. It follows immediately from Theorem \ref{ifonethenall}:

\begin{corollary}
If $\texttt{coTHEOREMS}$ has dense hard sequences, then $\texttt{TAUT}$ has dense hard sequences and is not easy on average (in $\textbf{AvgP}$) for a balanced distribution. 
\end{corollary}

The next section provides an argument why we should expect $\texttt{TAUT}$ to have dense hard sequences and which ones they are.
\section{Density of True Unprovable Sentences}\label{densityunprovablesection}
G{\"o}del's First Incompleteness Theorem's shows that undecidable sentences exist for any sound, consistent theory, which raises the question of how prevalent they are. This section will show that the true unprovable (for any size proof) sentences predicted by Chaitin's version of the Incompleteness Theorem are dense, following Calude and J{\"u}rgensen\cite{CaludeJurgensen}, and extend that result. It then formulates a conjecture that sequences associated with these true unprovable sentences (for proofs of any size) are dense hard sequences for $\texttt{coTHEOREMS}$ (for proofs of size $\leq t$). The next section shows that this conjecture implies that $\texttt{TAUT}$ has dense hard sequences.

Chaitin\cite{ChaitinIncompleteness} showed that although most strings have high Kolmogorov complexity, no sound, consistent, finitely axiomatized theory can prove any string has Kolmogorov complexity above a threshold.\footnote{Following the literature, we will use this term rather than Chaitin complexity, although the result builds on Chaitin's work.\cite{ChaitinIncompleteness} See Li and Vitanyi\cite{LiVitanyiBook} Section 2.7.1.} A counting argument due to Calude and J{\"u}rgensen\cite{CaludeJurgensen} shows that the set of sentences stating ``string $x$ has high Kolmogorov complexity'' is not only infinite but dense, because the set of strings with high Kolmogorov complexity is itself dense. 

Fix a universal TM $U$. Strings are encoded in binary (essential for the proof); let $|x|$ be the length of string $x$. Define the plain Kolmogorov complexity $C(x)=\min\{|p|:U(p)=x\}$.\footnote{For more background on Kolmogorov complexity, see Li and Vitanyi\cite{LiVitanyiBook} and Calude\cite{CaludeInfoRandomness}.} Define the set of random strings as $R_C=\{x:C(x)\geq \frac{|x|}{2}\}$.\footnote{An extensive literature has examined the power of random strings, for instance, Hirihara\cite{Hirahara}. Other definitions of random strings that have been used in the literature would work equally well for this paper's results.} Let $r_x$ be a sentence in the language of a sound, consistent, finitely-axiomatized theory $\mathcal{T}$ that encodes ``$x\in R_C$''. By Chaitin's Incompleteness Theorem, theory $\mathcal{T}$ cannot prove $r_x$ for $|x|$ sufficiently large.\cite{ChaitinIncompleteness} Calude and J{\"u}rgensen show:\footnote{See the proof of Theorem 5.2.} 

\begin{theorem}\label{phiupperdensity}
The set $H=\{i|\phi_i=r_x\}$ has positive upper density.
\end{theorem}
\begin{proof}
We have $|r_x|=|x|+c$ for some $c$, where $c$ is the overhead in bits to formalize arithmetically that string $x$ is Kolmogorov random. The essence of the proof is that we can replace $x$ with any other Kolmogorov random string of the same length. Nearly all strings of length $|x|$ are random as $|x|$ grows large (a share of $1-2^{-\frac{|x|}{2}}$), by a counting argument.\footnote{There are only $2^{\frac{|x|}{2}-1}$ possible short strings, so most strings of length $|x|$ cannot be compressed.} Therefore, there are nearly $2^{|x|}$ random strings to fill in the length $|x|$ slot in $r_x$. Then, the share of length $n$ sentences that are true but unprovable is bounded below by a number slightly less than $2^{-c}$.
\end{proof}
By a similar argument, the set of arithmetic sentences formalizing ``string $x$ has non-zero length'', which are true and provable, has positive upper density.  Downey et al\cite{Downeyetal} refer to such sets as being partially computable at density $r$.

We now use Theorem \ref{phiupperdensity} to motivate a conjecture about dense hard sequences. Just as $\mathcal{T}$ has difficulty with sentence $\phi_i=r_x$ (it cannot prove it for $\frac{|x|}{2}$ sufficiently large), we similarly might expect that nondeterministic $M$ has difficulty (requires superpolynomial time) with $\langle\phi_i,1^t\rangle$ where $\phi_i=r_x$ for $\frac{|x|}{2}$ sufficiently large.\footnote{We specify that $M$ is nondeterministic here and below, which will facilitate discussion of propositional proof systems in the next section.} That is, we might expect that $H$ is a blind spot for $M$ and not only $\mathcal{T}$.

Define set $H'$ to include $H$ as well as any other dense sets of true unprovable sentences in $\mathcal{T}$. $H'$ includes sentences $r_x$ for any universal TM $U$ used to define $C(.)$ and $R_C$. The set of $U$ that are universal is noncomputable because ``is universal'' is a nontrivial property, so by Rice's Theorem $H'$ is noncomputable. Furthermore, we define $H'$ to include dense sets relative to  $U$ with an arbitrary iterations of the Turing jump (i.e., adding an oracle for halting). 

\begin{conjecture}\label{superpolynomiallyhard}
For any nondeterministic $M$ accepting $\texttt{coTHEOREMS}$, there is a dense subset $H'_M$ of $H'$ for which $i\in H'_M$ if and only if $(\langle\phi_i,1^t\rangle)_{t\in \mathbb{N}}$ is a hard sequence for $M$.
\end{conjecture}

Conjecture \ref{superpolynomiallyhard} implies that  $\texttt{coTHEOREMS}$ has dense hard sequences. Note that $H'$ and therefore the term ``balanced'' may not be well defined---there may be dense sets of true unprovable sentences beyond those identified above, and which are inherently difficult to describe. This works in our favor, as the misbehavior of $H'$ strengthens the rationale that no $M$ could accept $\texttt{coTHEOREMS}$ in polynomial time.\footnote{The literature has studied the statement ``there is no derivation of contradiction of length $n$ from the axioms of $\mathcal{T};$\cite{PudlakFiniteDomain} statements related to $H'$ would also be of interest.}

The following stronger conjecture states that $M$'s blind spot about $H'$ is extremal:
\begin{conjecture}\label{exponentiallyhard}
Given nondeterministic $M$ accepting $\texttt{coTHEOREMS}$, there exists a dense subset $H'_M$ of $H'$, such that for any sufficiently large $|x|$ which depends on $M$ and $\mathcal{T}$, $M$ requires $2^t$ steps for all $t$ on any input $\langle\phi_i,1^t\rangle$ where $i\in H'_M$.
\end{conjecture}
Conjecture \ref{exponentiallyhard} has a similar form to Chaitin's Incompleteness Theorem, with an absolute upper threshold for the size of strings beyond which $M$ has no information at all, and must apply brute force for all $t$.
\section{Dense Hard Sequences for Tautologies}\label{tautsection}
The section discusses Conjecture \ref{superpolynomiallyhard}'s implications for propositional proof systems.\footnote{There is an extensive literature considering families of tautologies that may be hard, surveyed in Kraj\'{\i}\v{c}ek\cite{Krajicekproof}.}

\begin{theorem}
If Conjecture \ref{superpolynomiallyhard} holds, then: (i) $\texttt{TAUT}$ has dense hard sequences for nondeterministic TMs and for proof systems; (ii) it is not easy on average for a nondeterministic TM to recognize tautologies, i.e., $\texttt{TAUT}\notin\textbf{AvgP}$; and (iii)~in particular, for any TM $M$ (propositional proof system $Q$), a dense subset of tautologies which encode ``there is no $\mathcal{T}$ proof of size $\leq t$ shows that string $x$ is Kolmogorov random'' are hard sequences for $M$ (for $Q$).
\end{theorem}
\begin{proof}
(i) If Conjecture \ref{superpolynomiallyhard} holds, then $\texttt{coTHEOREMS}$ has dense hard sequences for nondeterministic TMs. By Theorem \ref{ifonethenall}, equivalently $\texttt{TAUT}$ has dense hard sequences for nondeterministic TMs. (ii) This follows from Theorem \ref{notinAvgP}. (iii) Choose a $p$-isomorphism $f$ mapping $\texttt{coTHEOREMS}$ to $\texttt{TAUT}$. By the conjecture, for a dense set of $i\in H'$, the sequences of tautologies  $(f(\langle\phi_i,1^t\rangle))_{t\in\mathbb{N}}$ are dense hard sequences for nondeterministic TMs and encode ``there is no $\mathcal{T}$ proof of size $\leq t$ shows that string $x$ is Kolmogorov random''. They are also dense hard sequences for proof systems.
\end{proof}

The above statement is stronger than a folklore informal conjecture in the literature that most tautologies hard for any given proof system, as this conjecture is focused on individual tautologies and not families.\footnote{Kraj\'{\i}\v{c}ek\cite{Krajicekproof} states that this has been a folklore conjecture since at least Chv{\'a}tal and Szemer{\'e}di\cite{ChvatalSzemeredi}.}

It is possible that $\texttt{TAUT}$ has dense hard sequences that can be described by $p$-isomorphism with $\texttt{coTHEOREMS}$, but Conjecture \ref{superpolynomiallyhard} does not hold. It is also possible that $\texttt{TAUT}$ has dense hard sequences, but these cannot be described by $p$-isomorphism with $\texttt{coTHEOREMS}$. The value added of Conjecture \ref{superpolynomiallyhard} is that it provides information on which are the dense hard sequences of tautologies,\footnote{Pich and Santhanam\cite{PichSanthanam} explore the hardness of proving tautologies related to random strings relative to a time-bounded variant of Kolmogorov complexity.} why no optimal propositional proof system exists, and which distributions are balanced. 

Conjecture \ref{superpolynomiallyhard}, if true, suggests that the nonexistence of an optimal propositional proof system is inherently a manifestation of undecidability, as that conjecture employs a dense set of undecidable sentences. This could be formalized as follows:

\begin{conjecture}\label{zfcconjecture}
For any propositional proof system $Q$, there is a conservative extension $\mathcal{T}$ of ZFC (or PA) that outperforms $Q$ in proving tautologies. Adding an undecidable statement as a new axiom to $\mathcal{T}$ to create a conservative extension $\mathcal{T'}$ further improves upon $\mathcal{T}$ in proving tautologies.
\end{conjecture}

We suspect that Conjecture \ref{superpolynomiallyhard} implies  Conjecture \ref{zfcconjecture} but can prove only this weaker implication: Under Conjecture \ref{superpolynomiallyhard}, for any $M$ that accepts $\texttt{coTHEOREMS}$, there is a hard sequence for $M$ corresponding to a true sentence which is unprovable in the theory $\mathcal{T}$ used to define $\texttt{coTHEOREMS}$. $M$ can be improved by creating a new TM $M'$ which checks for this hard sequence in its input and accepts, and otherwise runs $M$. This is analogous to adding an undecidable sentence as an axiom for a theory. However, it is different, as $M'$ outperforms $M$ on only one hard sequence. By contrast, adding an undecidable sentence in a theory as an axiom  has more far reaching effects of proving infinitely more theorems and shortening proofs to an arbitrary extent.\footnote{See Ehrenfeucht and Mycielski\cite{Ehrenfeucht}.} We might expect this new theory to perform significantly better as a propositional proof system, with smaller proofs on more than just one hard sequence.

\section{Conclusion and Further Research}\label{conclusion}
This paper has stated a conjecture linking the average-case hardness proving tautologies are theorems. It suggests numerous avenues for further research.

A key question is whether the above analysis is relevant to $\textbf{NP}$ and the existence of one-way functions (OWFs), which has been the driving force behind research on average-case complexity. In particular, do dense hard sequences or just hard sequences exist for some $\textbf{NP}$-complete languages and for OWFs? If OWFs exist, it is easy to generate hard instances, and therefore hard sequences exist for a given OWF. It is known that $\texttt{SAT}$ has no almost optimal algorithm if and only if it has hard sequences if and only if there is no optimal proof system for $\texttt{SAT}$, i.e., there are endless improvements on a satisfying assignment as certificate for $\texttt{SAT}$ (Beyersdorff et al\cite{BeyersdorffKoblerMessner}).

There are statements structurally similar to Conjecture \ref{superpolynomiallyhard} implying a resolution to other open problems. If there are sequences in a $\Pi^p_2$-complete language that are hard even for $M$ with an oracle for $\textbf{coNP}$, then the polynomial hierarchy ($\textbf{PH}$) does not collapse at the first level ($\Pi^p_i\neq\Pi^p_{i+1}$ for $i=1$). The reason is that for $M$ with an oracle for $\textbf{coNP}$, there are no hard sequences for a language in $\textbf{coNP}$. A series of similar statements for each $i$ imply no collapse at any level. By the argument in Monroe\cite{MonroeTCS}, the existence of hard sequences for some $\Pi^p_i$-complete language implies that some hard sequence takes the form of a statement about bounded nonhalting that refers to the noncomputable language $\Pi_i$, as in Theorem \ref{equivalentstatements} (viii) above. These link the noncollapse of the polynomial and arithmetic hierarchies. 

Statements about hard sequences are highly flexible tools, and it may be possible to ``axiomatize'' many beliefs about open questions into a similar form. As another example, $\textbf{NL}\neq\textbf{P}$ if for any nondeterministic $M$ accepting a $\textbf{P}$-complete language, there is a an $\textbf{L}$-uniform sequence that $M$ cannot accept within a logarithmic space bound. More generally, a statement ``any $M$ accepting language $L$ in class $\mathcal{C}$ has hard sequences'', we can adjust the enhancements to or constraints on $M$'s resources, the language $L$, the class $\mathcal{C}$, what ``hard'' means in defining a hard sequence, what resource constraint is used to define uniformity, how prevalent are the hard sequences (existence versus density), and why we think the hard sequences are hard like Conjecture \ref{superpolynomiallyhard}.

These examples suggest a \emph{noncomputability hypothesis}: for most open problems, some statement about hard sequences with a link to noncomputability implies a resolution to that problem.\footnote{This is analogous to and partly overlaps Pudl{\'a}k's Feasible Incompleteness Thesis---see \cite{PudlakFiniteDomain} and \cite{PudlakLogicalFoundations} Section 6.4.} This hypothesis is testable---we can try to come up with such statements or fail. For instance, we have thus far failed to identify a statement about hard sequences would imply that the Unique Games Conjecture holds. A better understanding of which open problems can and cannot be linked to a statement about hard sequences might reveal what tools are needed to resolve those problems.\footnote{Pudl{\'a}k\cite{PudlakLogicalFoundations} (p. 564) suggests there may be value in identifying conjectures that can serve as tentative axioms until their nature is identified as provable theorems, independent statements, or factually incorrect statements.}

Such conjectures also provide a tentative set of axioms describing a hypothetical world to explore. For instance, the existence of hard sequences allows us to reason about the set of all TMs accepting a language, which has been a key challenge for complexity theory. They create a partial order over all TMs, organizing them so as to provide clarity regarding the behavior of the set of all TMs.

A recurring question in the literature is whether hard sequences of inputs or tautologies can be constructed. Conjecture \ref{superpolynomiallyhard} is rooted in Chaitin's Incompleteness Theorem, which is inherently nonconstructive unlike G{\"o}del's First Incompleteness Theorem. However, there is a scenario in which a hard sequence could be constructed. Under Conjecture \ref{exponentiallyhard}, hard sequences are defined to require $2^t$ steps for all $t$. Given $M$, it is therefore possible to enumerate $\phi_i$ for which $M$ on sequence $\langle\phi_i,1^t\rangle$ does not require $2^t$ steps for all $t$---because for these sequences, $M$ requires less than $2^t$ steps for some $t$. The set of such $\phi_i$ is c.e. If the set is not only c.e. but also c.e.-complete, this would allow construction of a hard sequence since the complement would then be a productive set (Rogers \cite{Rogers}).

The set $H'$, which includes all dense sets of true unprovable sentences, deserves further study. We suspect that $H'$ is more complex and poorly-behaved than the statements above. Another interesting question is whether there are dense sets of true unprovable sentences that do not have Kolmogorov random strings embedded.

In sum, there are numerous areas for further research including the relevance of hard sequences for $\textbf{NP}$ and OWFs, developing conjectures implying a resolution of other open problems, elaborating the implications of these conjectures, exploring other applications for hard sequences, understanding whether hard sequences can be constructed, and exploring the nature of $H'$. 
\bibliographystyle{amsplain}
\bibliography{equivalence}

\end{document}